\documentclass{amsart}
\usepackage{amsmath,amsthm,amssymb,amscd}
\usepackage{graphicx}
\usepackage{tikz-cd}
\usepackage{paralist}
\usepackage{environ}
\usepackage{xcolor}
\usepackage{hyperref}
\usepackage{centernot}
\usepackage{mathtools}
\usepackage[export]{adjustbox}

\usepackage[framemethod=TikZ]{mdframed}

\usepackage{marvosym}

\usepackage{xcolor,cancel}

\usepackage{booktabs,multirow}

\usepackage{subfigure}
\renewcommand{\thesubfigure}{\thefigure.\arabic{subfigure}}
\makeatletter
\renewcommand{\p@subfigure}{}
\renewcommand{\@thesubfigure}{\thesubfigure:\hskip\subfiglabelskip}
\makeatother

\makeatletter
\def\square{\pst@object{square}}
\def\square@i(#1,#2)#3{{\use@par\solid@star\psframe[origin={#1,#2}](#3,#3)}}
\makeatother

\makeatletter
\DeclareFontFamily{U}{tipa}{}
\DeclareFontShape{U}{tipa}{bx}{n}{<->tipabx10}{}
\newcommand{\arc@char}{{\usefont{U}{tipa}{bx}{n}\symbol{62}}}%

\newcommand{\arc}[1]{\mathpalette\arc@arc{#1}}

\newcommand{\arc@arc}[2]{%
  \sbox0{$\m@th#1#2$}%
  \vbox{
    \hbox{\resizebox{\wd0}{\height}{\arc@char}}
    \nointerlineskip
    \box0
  }%
}
\makeatother

\makeatletter
\newcommand{\doublewedge}{\big@doubleop{\wedge}}
\newcommand{\big@doubleop}[1]{%
  \DOTSB\mathop{\mathpalette\big@doubleop@aux{#1}}\slimits@
}

\newcommand\big@doubleop@aux[2]{%
  \sbox\z@{$\m@th#1#2$}%
  \makebox[1.35\wd\z@][s]{$\m@th#1#2\hss#2$}%
}

\newcommand{\abs}[1]{\left|#1\right|}     

\theoremstyle{plain}
\newtheorem{theorem}{Theorem}
\newtheorem{lemma}{Lemma}

\newtheorem{remark}{Remark}
\newtheorem{definition}{Definition}

\newtheorem{example}{Example}

\newtheorem{axiom}{Axiom}

\usepackage[title]{appendix}

\makeatletter
\@namedef{subjclassname@2020}{\textup{2020} Mathematics Subject Classification}
\makeatother

\begin{document}

\title{Characteristics of Vibrating Systems having Time-Constrained Energy}

\author[E. Cui]{E. Cui}
\address{
Department of Electrical \& Computer Engineering,
University of Manitoba, WPG, MB, R3T 5V6, Canada
}
\email{cuie@myumanitoba.ca}

\author[J.F. Peters]{J.F. Peters}
\address{
Department of Electrical \& Computer Engineering,
University of Manitoba, WPG, MB, R3T 5V6, Canada and
Department of Mathematics, Faculty of Arts and Sciences, Ad\.{i}yaman University, 02040 Ad\.{i}yaman, Turkey,
}
\email{james.peters3@umanitoba.ca}


\date{}

\begin{abstract}
This paper introduces an axiomatic basis for measuring the energy characteristic of vibrating dynamical systems. 
The basic approach is to compare non-modulated vs. modulated waveforms in measuring energy during the vibratory motion $m(t)$ at time $t$ of moving object such as off-road vehicle oscillating movements recorded in an infrared (IR) video. Modulation of $m(t)$ is achieved either physically by adjusting the load on a spring system or geometrically by adjusting the frequency $\omega$ of the Euler exponential in  \boxed{m(t)e^{\pm j \omega t}dt}.   Expenditure of energy $E_{m(t)}$ by a system is measured in terms of the area bounded by the motion $m(t)$ waveform at time $t$.
\end{abstract}
%
\maketitle
\tableofcontents

\section{Introduction}

The collection of oscillatory parts of a vibrating dynamical system recorded in a single IR video frame, constructs a motion vector field~\cite[p.93]{Palis1982}. For a detailed view of the basics of IR images, see~\cite{IRpaper}. For a typical vector field portrait~\cite{Parilla2024vectorFields,Parilla2017vectorFields,Parilla2022vectorFields}, see Fig.~\ref{fig:vectorfield}. Each one of the motion field vectors in a vibrating system is recorded in an IR video frame like the one in Fig.~\ref{fig:bothframe}. Motion waveforms of vibrating dynamical system have varying cyclic oscillations~\cite{DeLeoYork2024,Feldman2011vibration}.

\begin{figure}[!ht]
	\centering
	\includegraphics[width=55mm]{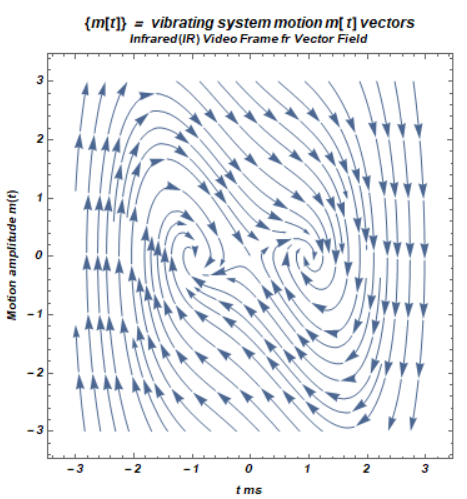}
	\caption{Sample IR frame vector field portrait~\cite{DeLeoYork2024}}
	\label{fig:vectorfield}
\end{figure}

\begin{figure}[!ht]
	\centering
	\includegraphics[width=0.75\textwidth]{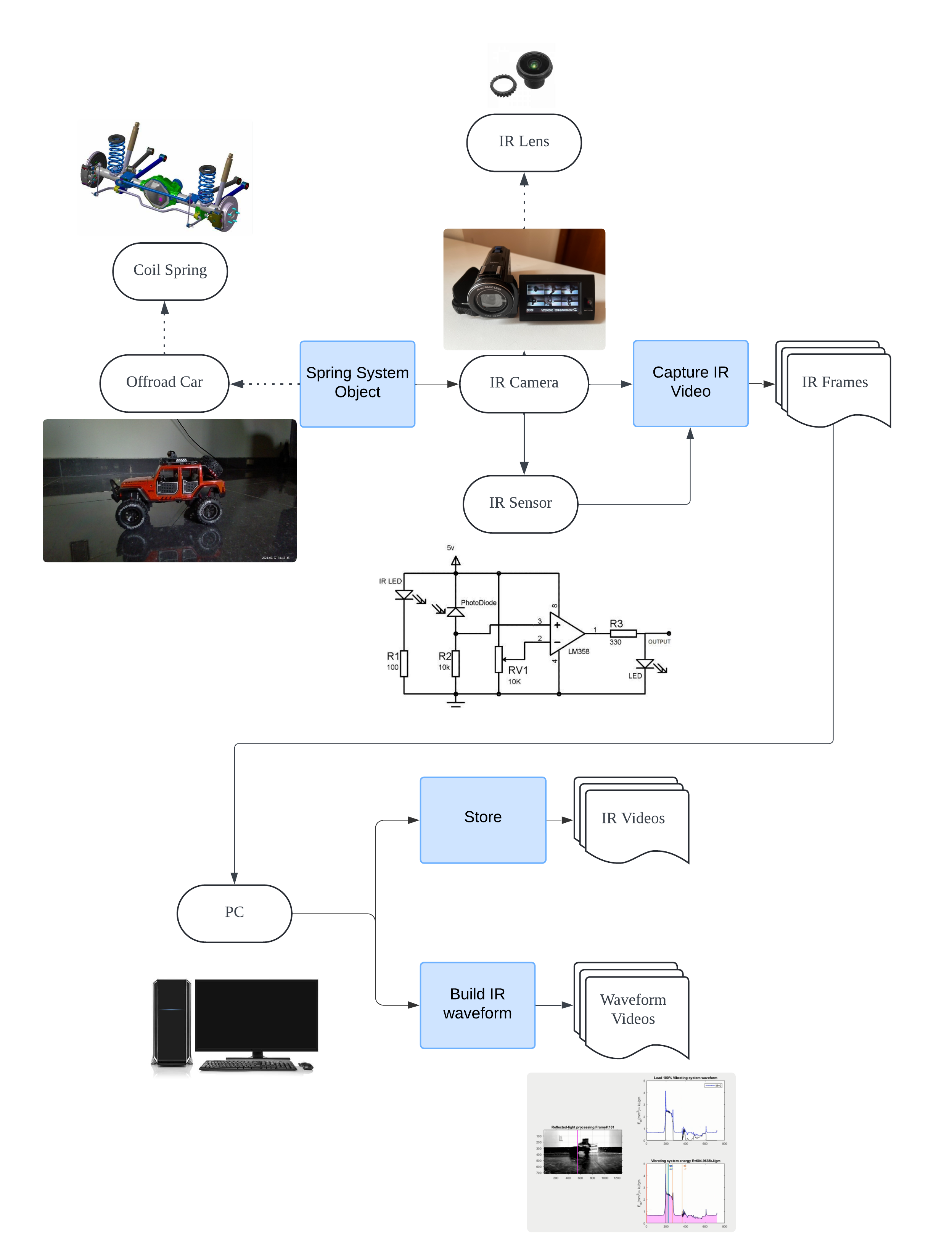}
	\caption{Flow chart of the IR spring system experiment}
	\label{fig:flowchartIR}
\end{figure} 

\begin{figure}[!ht]
	\subfigure[Loaded system with 100\% loads in frame 97]{
	\centering
	\includegraphics[width=0.4\textwidth]{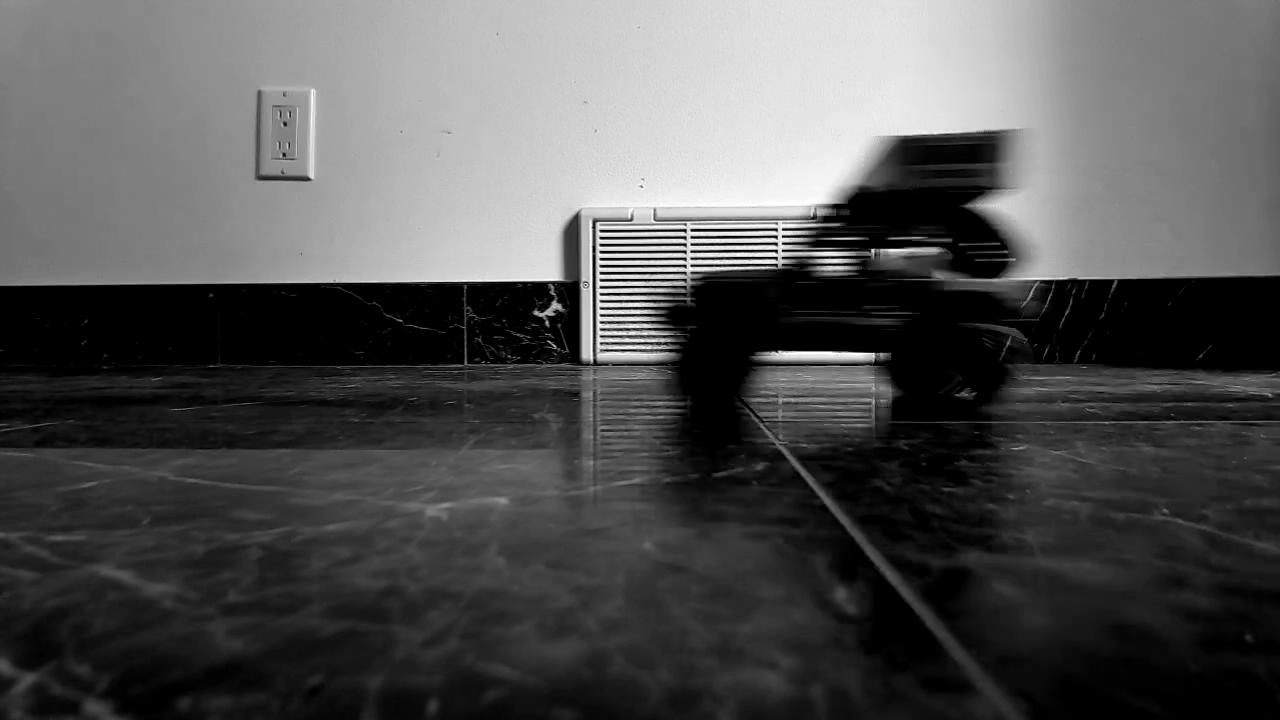}
	}\label{fig:case1frame}
	\subfigure[Spring system Without a Load in frame 89]{
	\centering
	\includegraphics[width=0.4\textwidth]{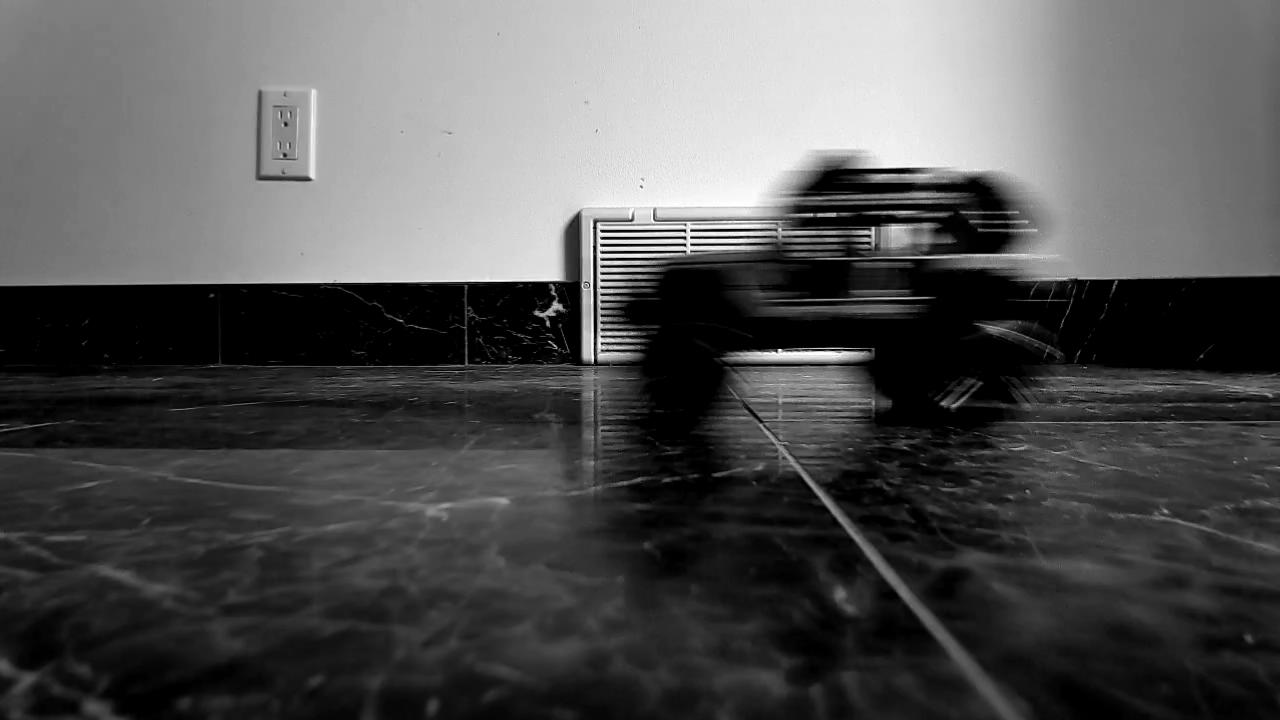}
	}\label{fig:case2frame}
	\caption{Loaded system and unloaded system IR videos}
	\label{fig:bothframe}	
\end{figure} 

\begin{figure}[!ht]
	\subfigure[Waveform of loaded spring system in frame 97]{
	\centering
	\includegraphics[width=0.4\textwidth]{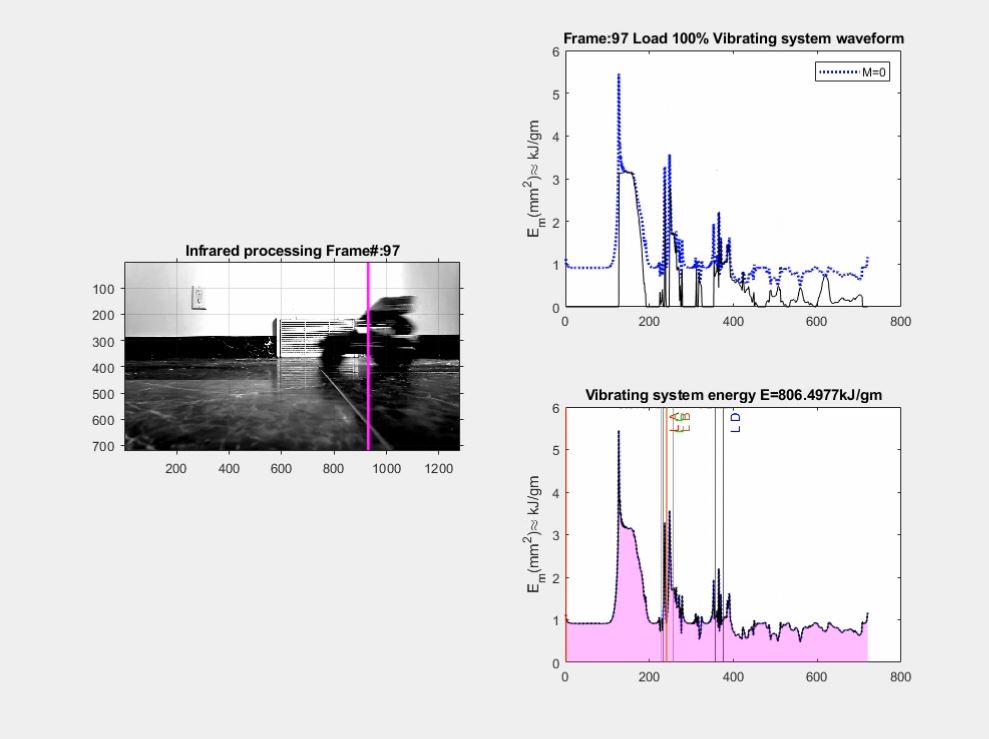}
	}\label{fig:case1}
	\subfigure[Modulated waveform of spring system without a load in frame 89]{
	\centering
	\includegraphics[width=0.4\textwidth]{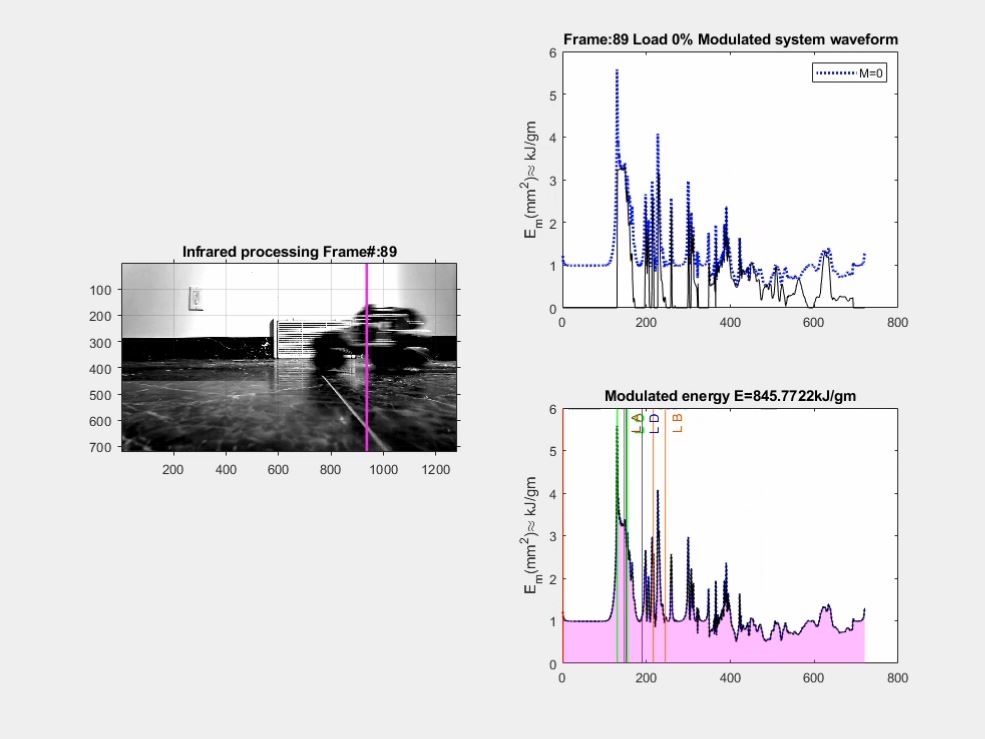}
	}\label{fig:case2}
	\caption{Loaded system waveform and modulated unloaded system waveform}
	\label{fig:caseall}	
\end{figure} 

\begin{figure}[!ht]

	
\end{figure}

\begin{table}[!ht]
    \centering
    \begin{tabular}{|l|l|l|l|l|l|l|l|l|l|}
    \hline
        Frames & Lobe A & Lobe B & Lobe C & Lobe D  \\ \hline
        95 & 85.93134179 & 50.61006225 & 18.14882145 & 21.86357325  \\ \hline
        96 & 85.93134179 & 34.72112252 & 18.14882145 & 21.86357325  \\ \hline
        97 & 85.93134179 & 42.19765703 & 18.14882145 & 21.86357325  \\ \hline
        98 & 85.93134179 & 35.33620105 & 18.14882145 & 21.86357325  \\ \hline
        99 & 85.93134179 & 50.61006225 & 18.14882145 & 21.86357325  \\ \hline
        100 & 85.93134179 & 50.61006225 & 18.14882145 & 21.86357325  \\ \hline
    \end{tabular}
		\caption{Case 1:  Physically modulated spring system energy}
		\label{table:case1}
\end{table}

\begin{table}[!ht]
		\centering
    \begin{tabular}{|l|l|l|l|l|l|l|l|l|l|}
    \hline
        Frames & Lobe A & Lobe B & Lobe C & Lobe D  \\ \hline
        86 & 76.23429076 & 28.33540646 & 21.84815468 & 14.94586226 \\ \hline
        87 & 76.23429076 & 24.91318173 & 21.84815468 & 14.94586226  \\ \hline
        88 & 76.23429076 & 28.33540646 & 20.51180804 & 14.94586226  \\ \hline
        90 & 76.23429076 & 23.16316581 & 21.84815468 & 14.94586226  \\ \hline
        91 & 76.23429076 & 28.33540646 & 21.84815468 & 14.94586226  \\ \hline
        92 & 76.23429076 & 28.33540646 & 21.84815468 & 14.94586226  \\ \hline
    \end{tabular}
		\caption{Case 2:  Euler exponentially modulated spring system energy}
    \label{table:case2}
\end{table}



This paper introduces an axiomatic basis for measuring the energy characteristic of vibrating dynamical systems.  For a motion waveform $m(t)$ at time $t$, measure of motion energy $E_{m(t)}$ is defined in terms of non-modulated energy $E_{nmod}$ and modulated energy $E_{mod}(t)$, i.e.,
\begin{center}
\vspace*{0.2cm}
\boxed{\boldsymbol{
E_{nmod(t)} = \int_{t_0}^{t_k}\abs{m(t)}^2dt 
}}
\end{center}
\vspace*{0.2cm}

Briefly, non-modulated energy $E_{nmod(t)}$ is modulated using the Euler exponential~\cite{Euler1748} $e^{j\omega t}$ to obtain modulated energy $E_{\mod(t)}$ with vary frequency $\omega$, i.e.,
\begin{center}
\vspace*{0.2cm}
\boxed{\boldsymbol{
E_{mod}(t) = \int_{t_0}^{t_k}\abs{m(t)e^{j\omega t}}^2dt. 
}}
\end{center} 
\vspace*{0.2cm} 

\begin{table}[!ht]\label{table:symbols}
\begin{center}
\begin{tabular}{|c|c|}
\hline
Symbol & Meaning\\ 
\hline\hline
$2^A$ & Collection of subsets of a nonempty set $A$\\
\hline
$A_i\in 2^A$ & Subset $A_i$ that is a member of $2^A$\\
\hline
$\mathbb{C}$ & Complex plane\\
\hline
$t$ & Instants Clock tick (cf. Atomic clock~\cite{Formichella2017})\\
\hline
$m(t)$ & Motion characteristic at time $t$\\
\hline
$\omega$ & Waveform Oscillation Frequency\\
\hline
$E_{m(t)}$ & Energy of motion $m(t)$ waveform at time $t$\\
\hline
$\varphi_t:2^A\to \mathbb{C}$ & $\varphi$ maps $2^A$ to complex plane $\mathbb{C}$ at time $t$\\
\hline
$\varphi_t(A_i\in 2^A)\in\mathbb{C}$ & $\varphi(A_i\in 2^A)\in\mathbb{C}$ (in complex plane at time $t$).\\
\hline
\end{tabular}
\caption{Principal Symbols Used in this Paper}
\end{center}
\end{table}
\vspace*{0.2cm}

\section{Preliminaries}
Highly oscillatory, non-periodic waveforms provide a useful portrait of vibrating system behavior.  Included in this paper is an axiomatic basis for measuring the energy characteristic of dynamical systems.  Briefly, a {\bf characteristic} is a mapping $\varphi_t:A_i\to \mathbb{C}$, which maps a subsystem $A_i$ in a system $A$ to a point in the complex plane $\mathbb{C}$.

\begin{definition}\label{def:system}{\bf(System)}.\\
A {\bf system} $A$ is a collection of interconnected components (subsystems $A_i\in 2^A$) with input-output relationships.
\qquad \textcolor{blue}{$\blacksquare$}
\end{definition}
\vspace*{0.2cm}

\begin{definition}\label{def:dynamicalSystem0}{\bf(Dynamical System)}.\\
A {\bf dynamical system} is a time-constrained physical system.
\qquad \textcolor{blue}{$\blacksquare$}
\end{definition}
\vspace*{0.2cm}

\begin{definition}\label{def:dSwaveform}{\bf(Dynamical System Output Waveform)}.\\
The output of a {\bf dynamical system} is a time-constrained sequence of discrete values.
\qquad \textcolor{blue}{$\blacksquare$}
\end{definition}
\vspace*{0.2cm}

A motion waveform is a graphical portrait of the radiation emitted by moving system (e.g., off-rode vehicle, walker, runner, biker) with oscillatory output.

\begin{figure}[!ht]
	\centering
	\includegraphics[width=65mm]{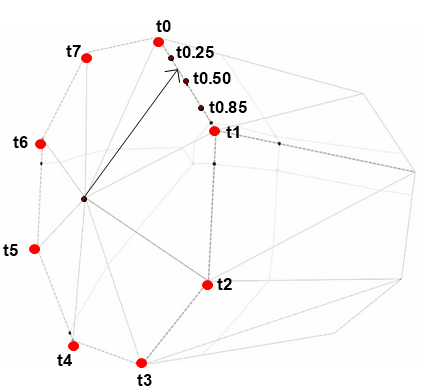}
	\caption{Geometric view of a Morse instants clock~\cite{Ludmany2023}, cf. precision of atomic clock~\cite{Formichella2017}}
	\label{fig:MorseClock}
\end{figure}

\begin{axiom}\label{axiom:clock}{\bf (Instants Clock)~\cite{JPTL2024}}.\\
Every system has its own instants clock, which is a cyclic mechanism that is a simple closed curve with an instant hand with one end of the instant hand at the centroid of the cycle and the other end tangent to a curve point indicating an elapsed time in the motion of a vibrating system.  A clock tick occurs at every instant that a system changes its state. 
\qquad \textcolor{blue}{$\blacksquare$}
\end{axiom}
\vspace*{0.2cm}

\begin{remark}{\bf (What Euler tells us about time)}.\\
On an instants clock, every reading $t\in \mathbb(C)$, a point \boxed{t = a+jb, a,b\in\mathbb{R}} in the complex plane.
For example, \boxed{t_{0.25}=0.25 + j0 = 0.25} in Fig.~\ref{fig:MorseClock}. The Morse instants clock is also called a homographic clock~\cite{HofmannKasner1928}, since the tip of an instant clock $t$-hand moves on the circumference a circle, where $t$ is a complex number~\cite{Kasner1928}. For $t$ at the tip of a vector with radius $r$, angle $\theta$ and \boxed{a = rcos\theta,b=rsin\theta} in the complex plane, then
\begin{center}
\boxed{\boldsymbol{
t = a+jb = rcos\theta + jrsin\theta = e^{j\theta}.
}}
\end{center}
\vspace*{0.2cm}
An instant of time viewed as an exponential is inspired by Euler~\cite{Euler1748}. \qquad \textcolor{blue}{$\blacksquare$}
\end{remark}
\vspace*{0.2cm}

\begin{example}
A sample Morse instants clock is shown in Fig.~\ref{fig:MorseClock}. The clock hand points to the elapsed time in milliseconds (ms) after a system has begun vibrating. The clock face is represented by a polyhedral surface in a Morse-Smale in a convex polyhedron in 3D Euclidean space~\cite{Ludmany2023}, since a Morse-Smale polyhedron is an example of a mechanical shape descriptor ideally suited as clock model because of its underlying piecewise smooth geometry. This form of an instants clock has been chosen to emphasize that each elapsed time $t_k$ (after initial time $t_0$) is a real number in an instants interval \boxed{\left[t_0,t_k\right]\in\mathbb{R}} in which the current time $t_k$ is indeterminate. That is, all elapsed times are approximate. From a planar perspective, the proximity of sets of instants clock times is related to results given for computational proximity in the digital plane~\cite{Peters2019}.  In this example, the instant hand is pointing to an elapsed time between $t_{0.25}$ ms and $t_{0.50}$ ms. Several examples of time-constrained dynamical systems motion recorded in video frames on a coarse-grained 24 hour clock are shown in~\cite[Fig. 9]{Cui2022}.
\qquad \textcolor{blue}{$\blacksquare$}
\end{example}
\vspace*{0.2cm}

\begin{example}
Typical examples of dynamical systems are
\begin{compactenum}[1$^o$]
\item Pendulum system $A$ with a waveform $\varphi(A) = -sin(t)$~\cite{DeLeoYork2024} having a phase portrait with loop nodes $C_0, C_1$, saddle points $s_0,s_1$ on curves $g_0,h_0,g_1,h_1$ like those shown in Fig.~\ref{fig:asymptotic}.
\item Off-road vibrating system radiation with a pair of waveform portraits record in a sequence of IR video frames (see, e.g., Fig.~\ref{fig:case1}).  Let $\varphi(Head)$ be the vibration amplitude of off-road vehicle .  The upper waveform exhibits a dip in the amplitude that portrays a runner vibration, which is non-modulated, with low at 0.04 and high at 0.65 at time $t$.  Let 
\boxed{\boldsymbol{\varphi(m(t))e^{j\omega t}}} represent the off-road spring system vehicle spring system modulated waveform $m(t)$ with an Euler exponential at time $t$. The lower waveform exhibited in Fig.~\ref{fig:case2} gives a sample waveform energy, which is the spring system equivalent of modulating the motion waveform $m(t)$ with the Euler exponential. \qquad \textcolor{blue}{$\blacksquare$}
\end{compactenum}
\end{example}
\vspace*{0.2cm}

\begin{definition}\label{def:characteristic}{\bf (Clocked Characteristic of a subsystem)}.\\
The clocked characteristic of a subsystem $A_i$ of a system $A$ at time $t$ is a mapping \boxed{\varphi_t:A_i\in 2^A\to \mathbb{C}} defined by
$\varphi_t(A_i)=a+bj\in\mathbb{C}, a,b\in\mathbb{R}, j = \sqrt(-1),t\in\mathbb{R}$.
\qquad \textcolor{blue}{$\blacksquare$}
\end{definition}
\vspace*{0.2cm}

\begin{mdframed}[backgroundcolor=green!15]
\label{axiom:+ve char}{\bf  Axiomatic View of a Dynamical System Motion Characteristic.}

\begin{axiom}\label{axiom:complexNo}{\bf (Subsystem Motion Characteristic)}.\\
Let $A_i\in 2^A$ (subsystem $A_i$ in the collection of subsystems $2^A$ in system $A$) that emits changing radiation due to system movements (motion) and let $t$ be an elapsed time clock tick.  The motion characteristic of subsystem motion $A_i\in 2^A$ is a mapping $m_t:A_i\to \mathbb{C}$ defined by
\begin{center}
\boxed{\boldsymbol{ 
m_t(A_i)=a+bj\in \mathbb{C},a,b\in \mathbb{R}, j = \sqrt(-1),t\in\mathbb{R}.
}}
\vspace*{0.2cm}
\end{center} 
i.e., a subsystem $A_i$ motion characteristic of a system $A$ is a mapping $m_t(A_i\in 2^A)\in\mathbb{C}$ at time $t$. 
\qquad \textcolor{blue}{$\blacksquare$}
\end{axiom}
\vspace*{0.2cm}

\begin{remark}
For the motion characteristic, we write \boxed{m(t)} when it is understood that motion is on a subset $A_i\in 2^A$ in a dynamical system $A$.  Axiom~\ref{axiom:complexNo} is consistent with the view~\cite[p. 81]{Blair1976} of the characteristic vector field, represented here with a planer characteristic vector field $\boldsymbol{\xi}$ of a dynamical system with points $p(x,y,t)\in \xi$ that has positive complex characteristic coordinates at clock tick (time) $t$ such that\\
\begin{center}
\boxed{
\varphi_t(A_i\in 2^A) = 
p\in \xi = 
(a-jb)\frac{\partial \xi}{\partial x}
+ (a-jb)\frac{\partial \xi}{\partial y}
+ (a-jb)\frac{\partial \xi}{\partial t}, a,b\in\mathbb{R}.
}
\vspace*{0.2cm}
\end{center}
The 1-1 correspondence between every point $p$ having coordinates in the Euclidean space and points in the complex plane is lucidly introduced by D. Hilbert and S. Cohn-Vossen~\cite[\S 38, 263-265]{Hilbert1952}.  
\qquad \textcolor{blue}{$\blacksquare$}
\end{remark}
\end{mdframed}

\begin{figure}[!ht]
	\centering
	\includegraphics[width=90mm]{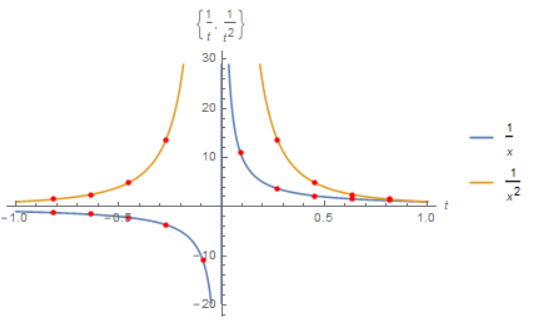}
	\caption{Subsystem characteristic asymptotic rate-of-change $\mathop{limit}\limits_{t\to \infty} \frac{\partial \varphi(A_i\in 2^A)}{\partial t}\to 0$.}
	\label{fig:asymptotic}
\end{figure} 

\vspace*{0.2cm}

\begin{example}
The rates-of-change of $\boldsymbol{f(t)=\frac{1}{t},g(t)=\frac{1}{t^2}}$ approach but do not reach 0 as $t\to \infty$, since $f(t)$ and $g(t)$ approach 0 asymptotically in Fig.~\ref{fig:asymptotic}.
\qquad \textcolor{blue}{$\blacksquare$}
\end{example}
\vspace*{0.2cm}

\begin{remark}
The heat radiated from a vibrating spring system is recorded with IR circuit (for the details on an IR sensor works, see Appendix~\ref{ap:IRcircuit}). The important result of a sequence of dynamical system oscillatory movements in IR video frames (taken together) is motion vector fieldd that resembles the tubular flow depicted in\cite[p.98]{Palis1982}.
\qquad \textcolor{blue}{$\blacksquare$}
\end{remark}

\begin{figure}[!ht]
\includegraphics[width=0.35\textwidth]{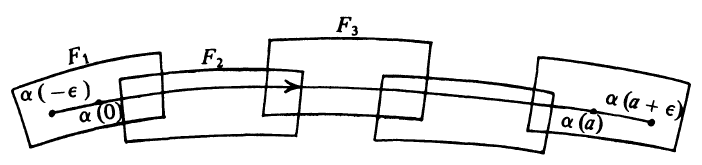}
\caption[]{Sequence of vibrating spring vector fields}
\label{fig:sysPair}
\end{figure}

\begin{figure}[!ht]

	\subfigure[Frame 1,K,N]{
	\centering
	\includegraphics[width=0.2\textwidth]{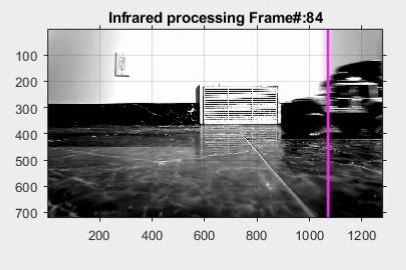}
	\includegraphics[width=0.2\textwidth]{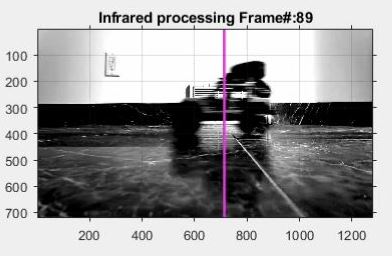}
	\includegraphics[width=0.2\textwidth]{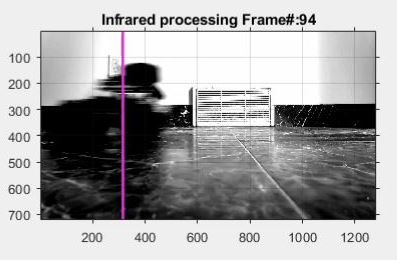}
	}\label{fig:frameset}
	\subfigure[Waveform 1,K,N]{
	\centering
	\includegraphics[width=0.2\textwidth]{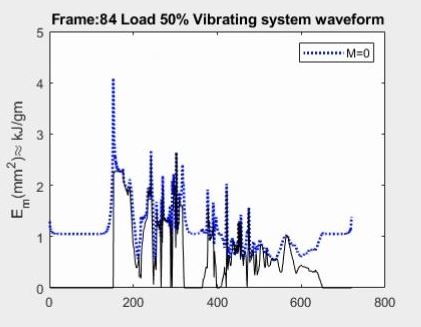}
	\includegraphics[width=0.2\textwidth]{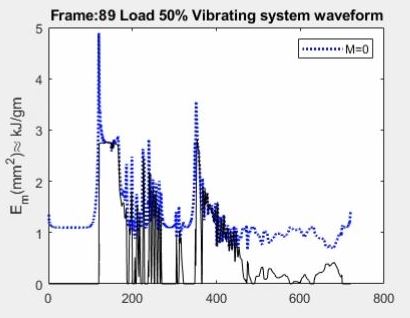}
	\includegraphics[width=0.2\textwidth]{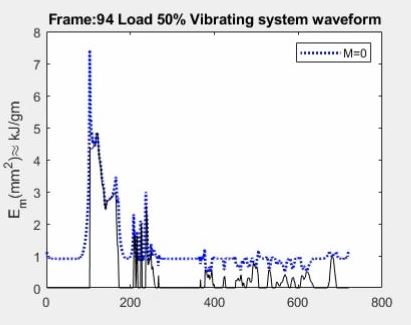}
	}\label{fig:waveformset}
	\caption{runner IR waveform over 1-to-n frames in a video.}
	\label{fig:after-n-frames}
\end{figure}

\begin{mdframed}[backgroundcolor=green!15]
\label{axiom:clocked char}{\bf  Vibrating System Time-Constrained Energy Characteristic.}
\begin{axiom}\label{axiom:waveformEnergy}{\bf (Waveform Energy)}.\\
A measure of dynamical system energy is the area of a finite planar region bounded by system waveform $m(t)$ curve at time $t$,
defined by
\begin{center}
\vspace*{0.2cm}
\boxed{\boldsymbol{
E_{m(t)} = \int_{t_0}^{t_1}\abs{m(t)}^2dt
}.}
\end{center}
\end{axiom}
\end{mdframed}
\vspace*{0.2cm}

\begin{remark}{\bf (Importance of Energy Characteristic Time-of-occurrence}.\\
From Axiom~\ref{axiom:waveformEnergy}, the energy characteristic is potentially different at every tick of the system.  Measurements of changing energy occurs on different measurement scales, depending on the accuracy of the measurement of the energy characteristic.  Evidence of this can be found in the endless variation of the radiation emanating from any vibrating physical system.  For more about the geomtry of camera-recorded system energy as a shape shifter, see~\cite{Peters2020}. 
\qquad \textcolor{blue}{$\blacksquare$}
\end{remark}
\vspace*{0.2cm}

\begin{lemma}\label{lemma:energy}
Dynamical system energy is time-constrained and is always limited.
\end{lemma}
\begin{proof}
Let $E_m$ be the energy of a dynamical system, defined in Axiom~\ref{axiom:waveformEnergy}. From Axiom~\ref{axiom:waveformEnergy}, system energy always occurs in a bounded temporal interval $\left[t_0,t_1\right]$.
Hence, $E_m$ is time constrained.
From Axiom~\ref{axiom:clock}, the length of a system waveform is finite, since, from Axiom~\ref{axiom:waveformEnergy}, system duration is finite. From Axiom~\ref{axiom:waveformEnergy}, system energy is derived from the area of a finite, bounded region.  Consequently, system energy is always finite.
\end{proof}
\vspace*{0.2cm}

\begin{theorem}\label{theorem:Em}
If $X$ is a dynamical system with waveform $m(t)$ at time $t$ and which changes with every clock tick, then observe
\begin{compactenum}[1$^o$]
\item System waveform characteristic values are in the complex plane.
\item\label{step:part1} System energy varies with every clock tick.
\item System radiation characteristics are finite.
\item All system characteristics map to the complex plane.
\item Waveform energy decay is a characteristic, which maps to $\mathbb{C}$ at clock tick t ms.
\end{compactenum}
\end{theorem}
\begin{proof}$\mbox{}$\\
\begin{compactenum}[1$^o$]
\item From Def.~\ref{def:characteristic}, a system characteristic is a mapping from a subsystem to the complex plane at time $t$,  From Axiom~\ref{axiom:complexNo}, every waveform motion characteristic $m(t)\in \mathbb{C}$ at time $t$, which is the desired result.
\item From Lemma~\ref{lemma:energy}, system energy is time-constrained and always occurs in a bounded temporal interval.  From Axiom~\ref{axiom:clock}, there is a new clock tick at every instant in time $t$ ms. From Axiom~\ref{axiom:waveformEnergy}, system energy varies with every clock tick.
\item From Axiom~\ref{axiom:clock}, all system radiation characteristics are finite, since system duration is finite.
\item\label{step:characteristic} From Axiom~\ref{axiom:complexNo}, every system $A$ characteristic is a mapping from a subsystem $A_i\in 2^A$ to the complex plane. From Axiom~\ref{axiom:waveformEnergy}, each instance of an energy characteristic value in $\mathbb{C}$ occurs a clock tick $t$, which is the desired result.
\item From the proof of step~\ref{step:characteristic}, the desired result follows.
\end{compactenum}
\end{proof}
\vspace*{0.2cm}

\begin{figure}[!ht]
\centering
\includegraphics[width=0.25\textwidth]{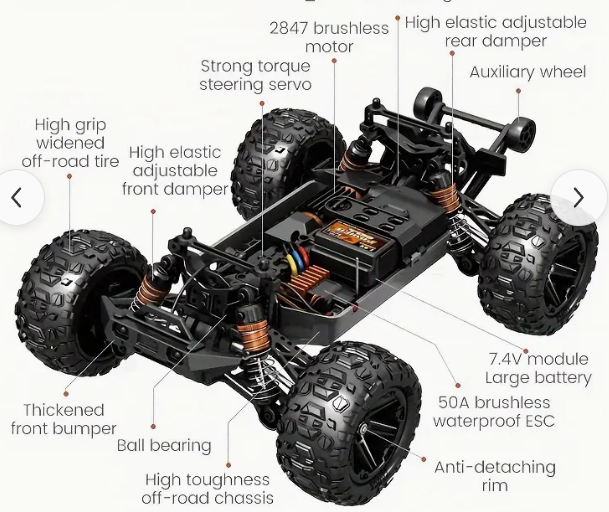}
\caption {Offroad spring interior design}
\label{fig:C2}
\end{figure}

\begin{figure}[!ht]
\centering
\includegraphics[width=0.25\textwidth]{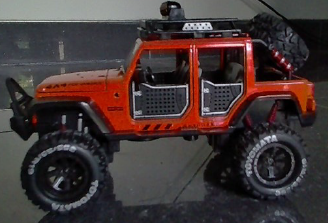}
\caption{Offroad spring system vehicle}
\label{fig:SV}
\end{figure}

\section{Conclusion}
An axiomatic approach to the characteristic energy commonly found in time-constrained vibrating dynamical waveforms introduced in this paper carries forward earlier work on dynamical systems introduced in~\cite{Cui2022,Tiwari2024}.  We have shown that the waveform energy decay characteristic is time-constrained and is always limited (see Lemma~\ref{lemma:energy} and part 5 of Theorem~\ref{theorem:Em}).

\section*{Acknowledgements}
We thank Tane Vergili, Surhabi Tiwari, Anna Di Concilio, Rob Alfano, Peter Smith, Arturo Tozzi and Sheela Ramanna for their insights concerning elements of the topology of dynamical systems.

This research has been supported by the Natural Sciences \&
Engineering Research Council of Canada (NSERC) discovery grant 185986 
and Instituto Nazionale di Alta Matematica (INdAM) Francesco Severi, Gruppo Nazionale per le Strutture Algebriche, Geometriche e Loro Applicazioni grant 9 920160 000362, n.prot U 2016/000036 and Scientific and Technological Research Council of Turkey (T\"{U}B\.{I}TAK) Scientific Human
Resources Development (BIDEB) under grant no: 2221-1059B211301223.

\noindent\textbf{Declaration of conflict of interest.} The authors declare that there is no conflict of interest.

\begin{appendix}\label{app}

\section{Sample off-road spring system}
	\indent The spring system in Fig.~\ref{fig:SV} helps reduces the impact caused by terrain for the off-road vehicles. It absorb and dampen the energy from impacts and vibrations caused by uneven terrain. This helps to maintain vehicle stability, enhance comfort, and protect the vehicle’s structure.The spring system as shown in Fig.~\ref{fig:C2} including 2 major parts, springs and dampers. Springs are the core components that absorb shock. They compress and expand to accommodate bumps and dips in the terrain. On the other hand, dampers control the rate at which the springs compress and extend. They prevent the vehicle from bouncing excessively after hitting a bump.

	Coil springs in Fig.~\ref{fig:SV} made of coiled steel and provide a balance between strength and flexibility. They compress under load and expand to absorb impacts. When an off-road vehicle encounters rough terrain, the springs compress to absorb the initial impact and dampers then control the rebound of the springs to prevent excessive oscillation, thereby reducing vibration and providing a smoother ride.

\section{IR Circuit}\label{ap:IRcircuit}
	\indent IR sensor is the key element for collecting Infrared video. IR sensor is a device that emits infrared radiations in order to sense object of the surroundings. These types of radiations are invisible to human eyes, but the receiver unit of IR sensor can detect these radiations.\newline
	There are two types of IR sensors: active and passive. The active IR sensor consists of both infrared source and infrared detector. The infrared source includes the infrared laser diode emits specific wavelength in the infrared spectrum. Infrared detectors include photodiodes or phototransistors receiving specific wavelength in the infrared spectrum. The passive IR sensor is basically infrared detectors, only detect specific wavelengths in the infrared spectrum.\newline
	\begin{figure}
	\centering
	\includegraphics[width = 0.75\textwidth]{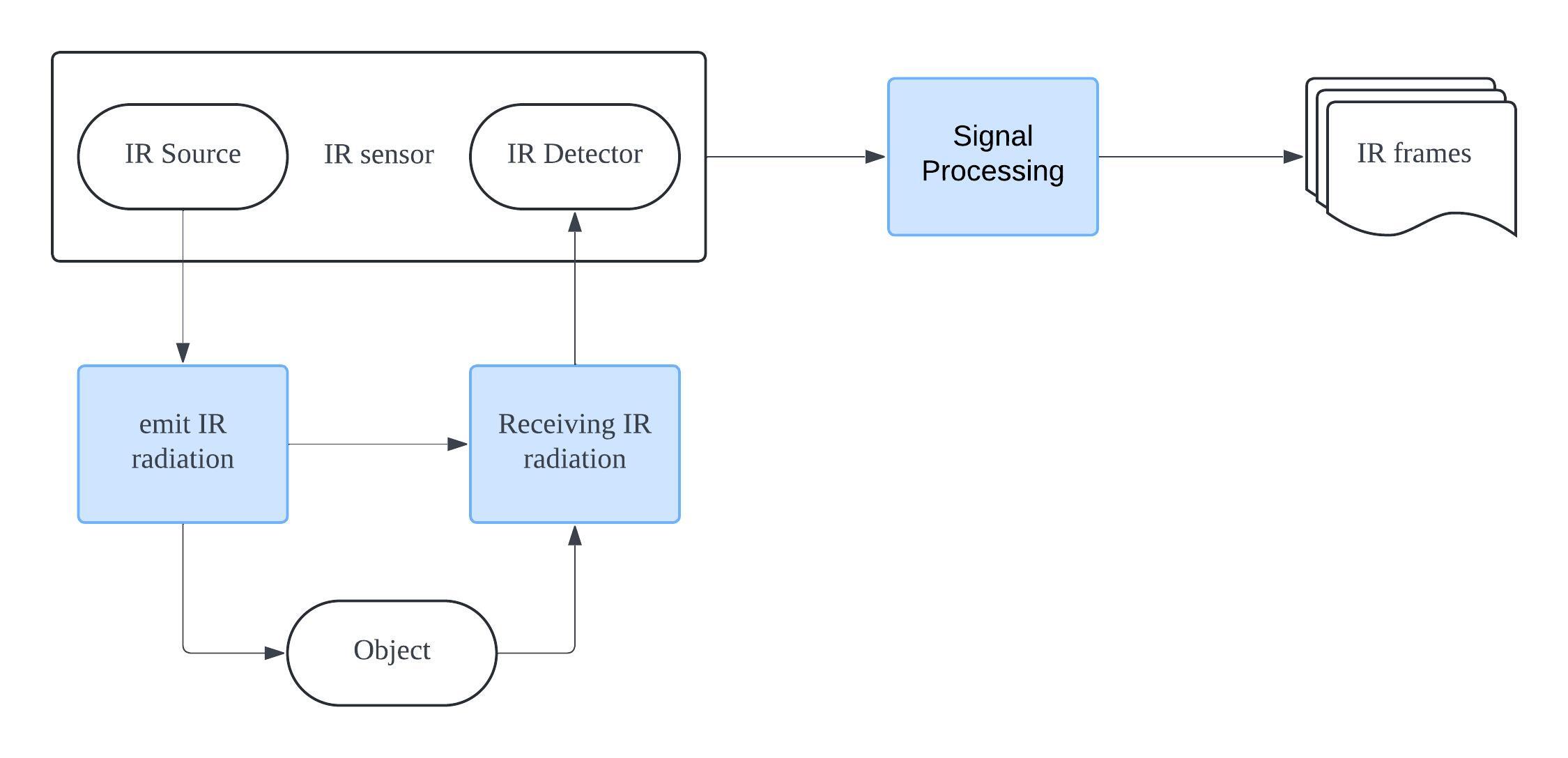}
	\caption {Infrared sensor flowchart}
	\label{fig:IRflowchart}
	\end{figure}
	The Fig.~\ref{fig:IRflowchart} illustrates the working principle of IR censor. In short, the energy emitted by the infrared source is reflected by an object and falls on the infrared detector. The output voltage then converts to IR video signal through signal processing.\newline
	\begin{figure}
		\centering
	\includegraphics[width = 0.75\textwidth]{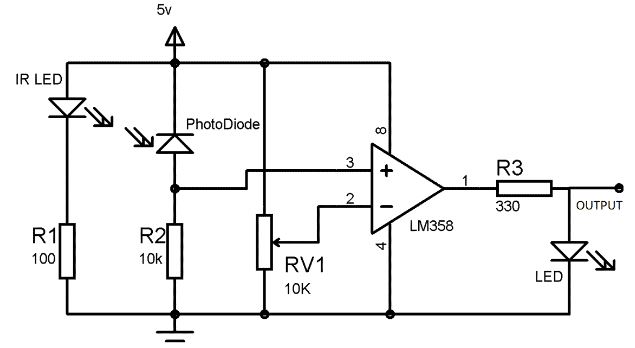}
	\caption {Infrared sensor circuit}
	\label{fig:IRcircuit}
	\end{figure}	
	In the Fig.~\ref{fig:IRcircuit}, The IR LED working as infrared source emitting infrared radiations. The photodiode receiving infrared radiation reflected from object. The amplifier controlling output voltage, it will change in proportion to the infrared radiations which the photodiode received. \newline
\vspace*{0.2cm}

\section{Comparisons between spring system-modulated and Matlab-modulated waveform}


\begin{figure}[!ht]

	\subfigure[Case1:Non-modulated waveform Energy Plot]{
	\centering
	\includegraphics[width=0.55\textwidth]{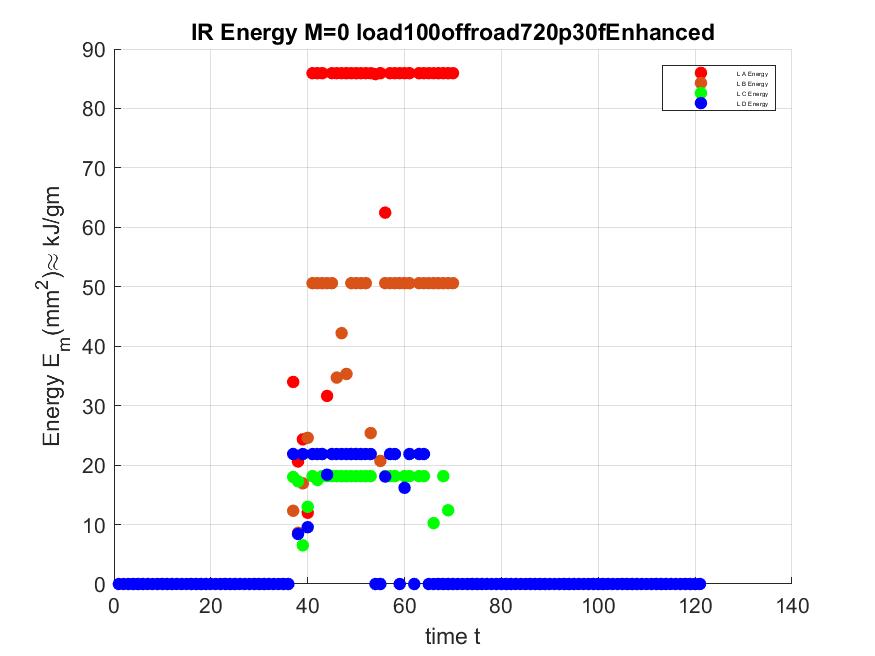}
	}\label{fig:energycase1}
	\subfigure[Case2:Modulated waveform Energy Plot]{
	\centering
	\includegraphics[width=0.55\textwidth]{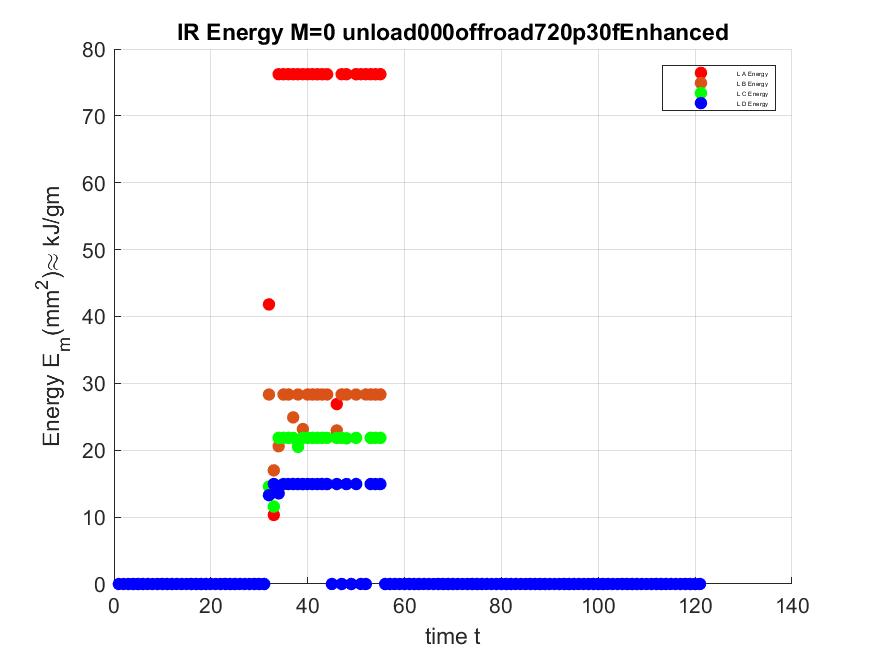}
	}\label{fig:energycase2}
	\caption{Energy Plot of Non-modulated Waveform and Modulated Waveform}
	\label{fig:energycaseall}
\end{figure}
\vspace*{0.2cm}
  The over all enerngy plot comparison between 2 case can be find in Fig.~\ref{fig:energycaseall}. There are 3 examples listed below indecates the begining, middle and the end of frame waveforms comparison.\\

\begin{example}\label{ex-a}
The waveform comparison for start of frames shown in Fig.~\ref{fig:exacaseall}. The vehicle here shown only the head in the frames and moving towards left. The energy table of each lobe in Non-modulated waveform see Table ~\ref{table:exacase1}.The energy table of each lobe in Modulated waveform see Table ~\ref{table:exacase2}. It is observed that the modulated unloaded waveform shared multiple similarities in shapes with the non-modulated loaded waveform. Compare the energy table of both modulated waveform and non-modulated waveform we can see the lobe energy of both waveform sharing similar level differences.  \\
\end{example}

\begin{figure}[!ht]

	\subfigure[Ex.~\ref{ex-a} Case1:Non-modulated waveform with 100\% loads]{
	\centering
	\includegraphics[width=0.4\textwidth]{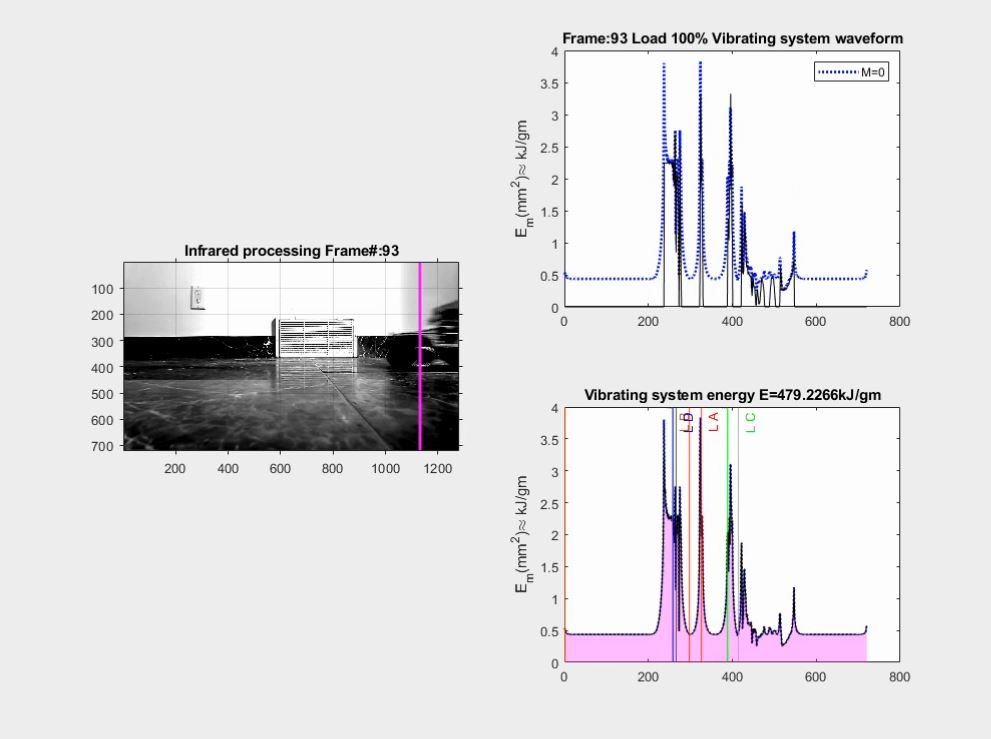}
	}\label{fig:exacase1}
	\subfigure[Ex.~\ref{ex-a} Case2:Modulated waveform with 0\% loads]{
	\centering
	\includegraphics[width=0.4\textwidth]{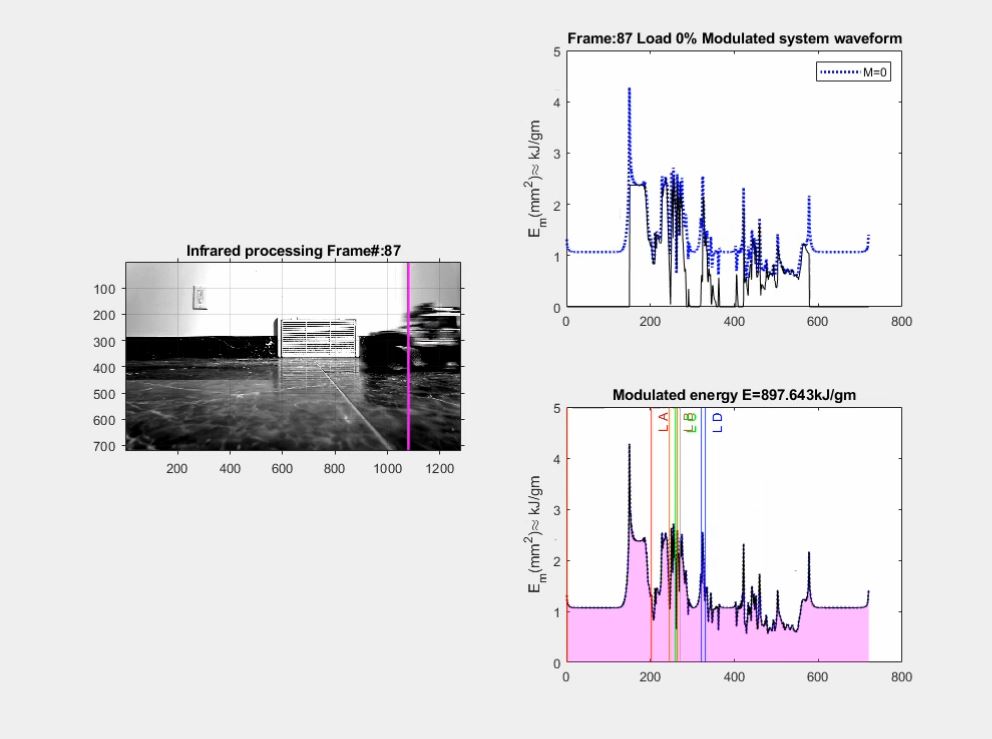}
	}\label{fig:exacase2}
	\caption{Ex.~\ref{ex-a} Non-modulated Waveform and Modulated Waveform}
	\label{fig:exacaseall}
	
\end{figure}

\begin{table}[!ht]
    \centering
    \begin{tabular}{|l|l|l|l|l|l|l|l|l|l|}
    \hline
        Frames & Lobe A & Lobe B & Lobe C & Lobe D  \\ \hline
        88 & 20.60368582 & 8.614353021 & 17.27595106 & 8.411250532  \\ \hline
        89 & 24.33440983 & 16.93526054 & 6.504560841 & 21.86357325  \\ \hline
        90 & 11.97927556 & 24.59000306 & 12.97974594 & 9.553969771  \\ \hline
        91 & 85.93134179 & 50.61006225 & 18.14882145 & 21.86357325  \\ \hline
        92 & 85.93134179 & 50.61006225 & 17.49519037 & 21.86357325  \\ \hline
        93 & 85.93134179 & 50.61006225 & 18.14882145 & 21.86357325  \\ \hline
    \end{tabular}
		\caption{Ex.~\ref{ex-a} Energy Table of Case 1: Non-modulated Waveform}
		\label{table:exacase1}
\end{table}

\begin{table}[!ht]
		\centering
    \begin{tabular}{|l|l|l|l|l|l|l|l|l|l|}
    \hline
        Frames & Lobe A & Lobe B & Lobe C & Lobe D  \\ \hline
        82 & 41.80832408 & 28.33540646 & 14.57766015 & 13.27935104  \\ \hline
        83 & 10.32399502 & 16.99477953 & 11.57665986 & 14.94586226  \\ \hline
        84 & 76.23429076 & 20.60028253 & 21.84815468 & 13.55562371  \\ \hline
        85 & 76.23429076 & 28.33540646 & 21.84815468 & 14.94586226  \\ \hline
        86 & 76.23429076 & 28.33540646 & 21.84815468 & 14.94586226  \\ \hline
        87 & 76.23429076 & 24.91318173 & 21.84815468 & 14.94586226  \\ \hline
    \end{tabular}
		\caption{Ex.~\ref{ex-a} Energy Table of Case 2: Modulated Waveform}
		\label{table:exacase2}
\end{table}

\begin{example}\label{ex-b}
The waveform comparison for middle of frames shown in Fig.~\ref{fig:exbcaseall}. The vehicle here shown whole parts in the frames and moving towards left. The energy table of each lobe in Non-modulated waveform see Table ~\ref{table:exbcase1}.The energy table of each lobe in Modulated waveform see Table ~\ref{table:exbcase2}. The result agrees with Ex.~\ref{ex-a}.\\
\end{example}
\vspace*{0.2cm}

\begin{figure}[!ht]

	\subfigure[Ex.~\ref{ex-b} Case1:Non-modulated waveform with 100\% loads]{
	\centering
	\includegraphics[width=0.4\textwidth]{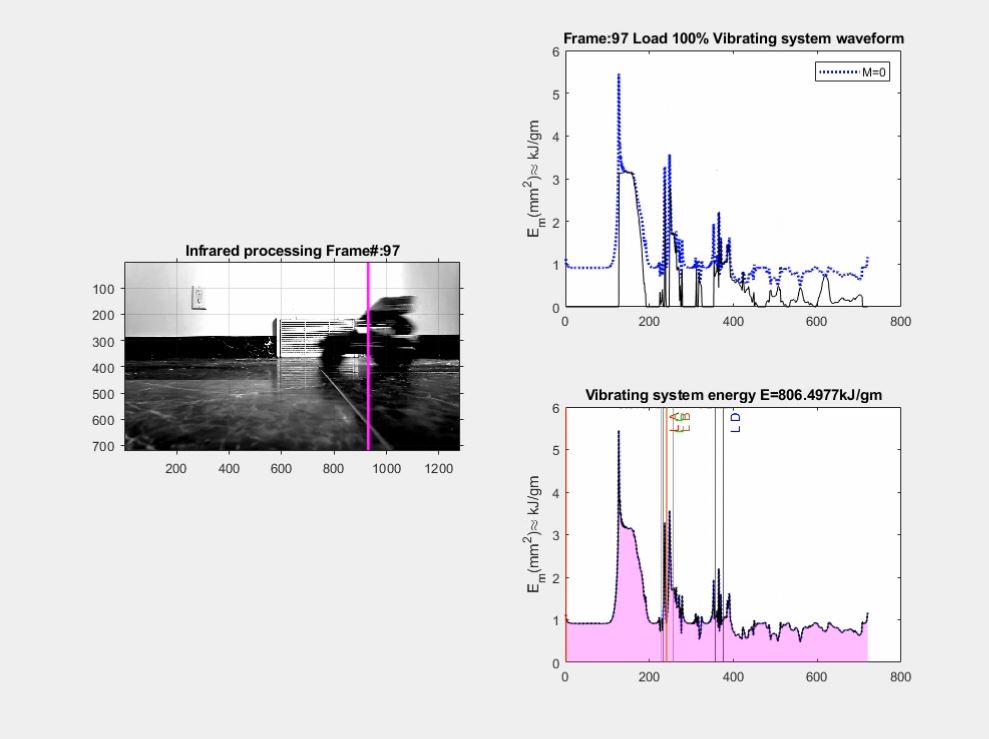}
	}\label{fig:exbcase1}
	\subfigure[Ex.~\ref{ex-b} Case2:Modulated waveform with 0\% loads]{
	\centering
	\includegraphics[width=0.4\textwidth]{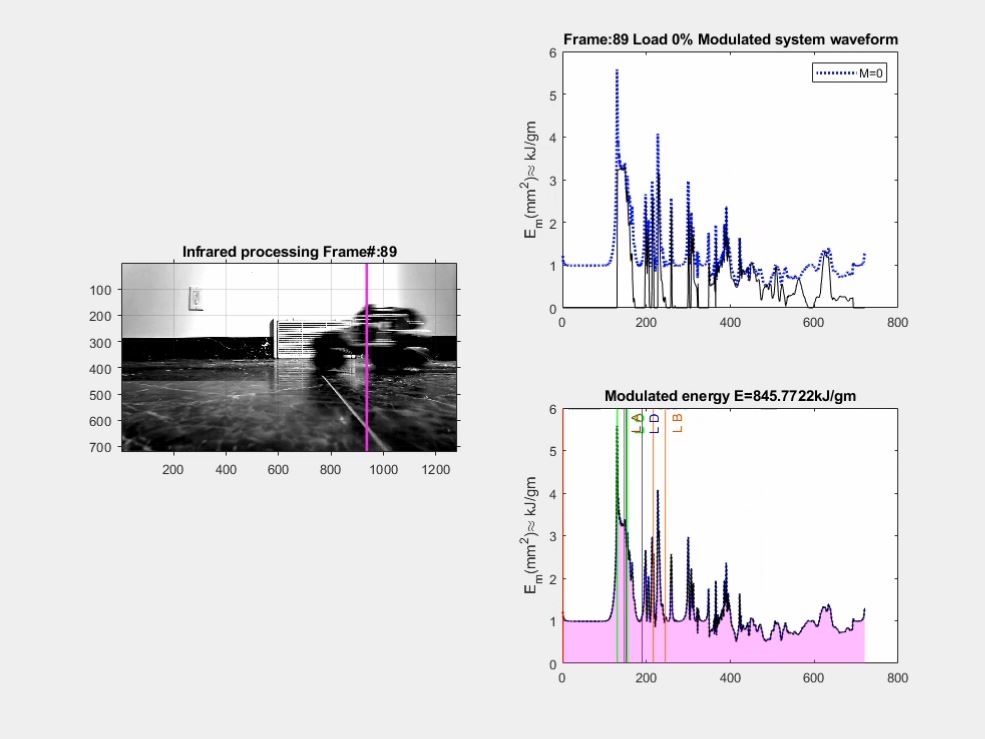}
	}\label{fig:exbcase2}
	\caption{Ex.~\ref{ex-b} Non-modulated Waveform and Modulated Waveform}
	\label{fig:exbcaseall}
	
\end{figure}

\begin{table}[!ht]
    \centering
    \begin{tabular}{|l|l|l|l|l|l|l|l|l|l|}
    \hline
        Frames & Lobe A & Lobe B & Lobe C & Lobe D  \\ \hline
        95 & 85.93134179 & 50.61006225 & 18.14882145 & 21.86357325  \\ \hline
        96 & 85.93134179 & 34.72112252 & 18.14882145 & 21.86357325  \\ \hline
        97 & 85.93134179 & 42.19765703 & 18.14882145 & 21.86357325  \\ \hline
        98 & 85.93134179 & 35.33620105 & 18.14882145 & 21.86357325  \\ \hline
        99 & 85.93134179 & 50.61006225 & 18.14882145 & 21.86357325  \\ \hline
        100 & 85.93134179 & 50.61006225 & 18.14882145 & 21.86357325  \\ \hline
    \end{tabular}
		\caption{Ex.~\ref{ex-b} Energy Table of Case 1: Non-modulated Waveform}
		\label{table:exbcase1}
\end{table}

\begin{table}[!ht]
		\centering
    \begin{tabular}{|l|l|l|l|l|l|l|l|l|l|}
    \hline
        Frames & Lobe A & Lobe B & Lobe C & Lobe D  \\ \hline
        86 & 76.23429076 & 28.33540646 & 21.84815468 & 14.94586226 \\ \hline
        87 & 76.23429076 & 24.91318173 & 21.84815468 & 14.94586226  \\ \hline
        88 & 76.23429076 & 28.33540646 & 20.51180804 & 14.94586226  \\ \hline
        90 & 76.23429076 & 23.16316581 & 21.84815468 & 14.94586226  \\ \hline
        91 & 76.23429076 & 28.33540646 & 21.84815468 & 14.94586226  \\ \hline
        92 & 76.23429076 & 28.33540646 & 21.84815468 & 14.94586226  \\ \hline
    \end{tabular}
		\caption{Ex.~\ref{ex-b} Energy Table of Case 2: Modulated Waveform}
		\label{table:exbcase2}
\end{table}

\begin{example}\label{ex-c}
The waveform comparison for end of frames shown in Fig.~\ref{fig:exccaseall}. The vehicle here left the frames, only the shadow of the vehicle in the frames and moving towards left. The energy table of each lobe in Non-modulated waveform see Table ~\ref{table:exccase1}.The energy table of each lobe in Modulated waveform see Table ~\ref{table:exccase2}. The result agrees with Ex.~\ref{ex-a} and Ex.~\ref{ex-b}.\\
\end{example}

\begin{figure}[!ht]

	\subfigure[Ex.~\ref{ex-c} Case1:Non-modulated waveform with 100\% loads]{
	\centering
	\includegraphics[width=0.4\textwidth]{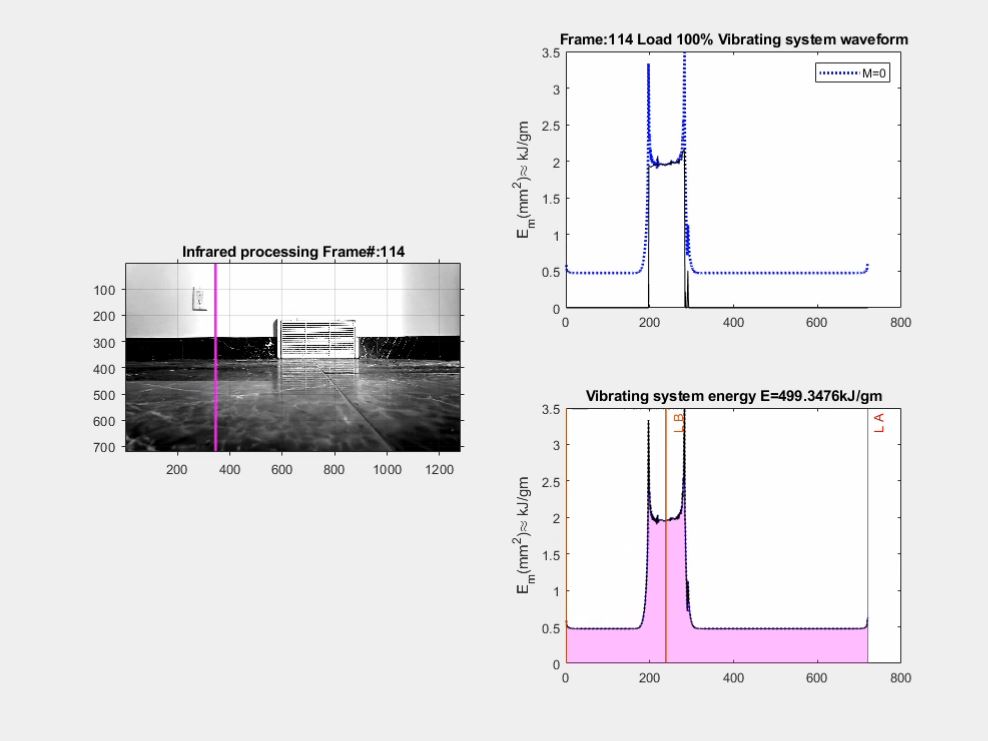}
	}\label{fig:exccase1}
	\subfigure[Ex.~\ref{ex-c} Case2:Modulated waveform with 0\% loads]{
	\centering
	\includegraphics[width=0.4\textwidth]{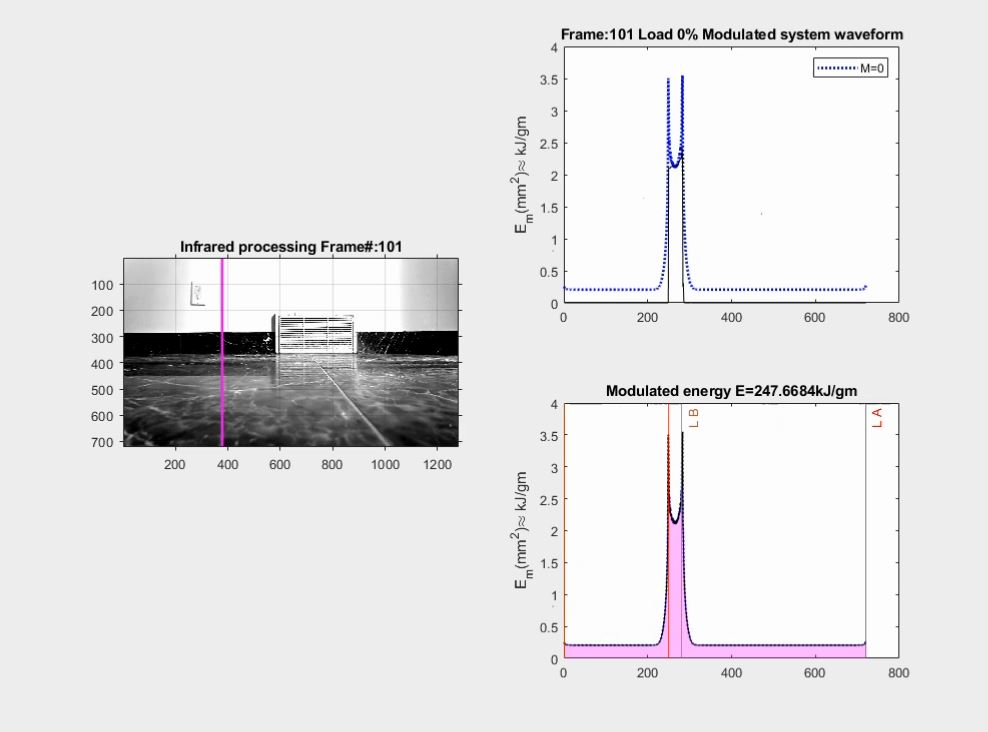}
	}\label{fig:exccase2}
	\caption{Ex.~\ref{ex-c} Non-modulated Waveform and Modulated Waveform}
	\label{fig:exccaseall}
	
\end{figure}

\begin{table}[!ht]
    \centering
    \begin{tabular}{|l|l|l|l|l|l|l|l|l|l|}
    \hline
        Frames & Lobe A & Lobe B & Lobe C & Lobe D  \\ \hline
        113 & 85.93134179 & 50.61006225 & 18.14882145 & 21.86357325  \\ \hline
        114 & 85.93134179 & 50.61006225 & 18.14882145 & 21.86357325  \\ \hline
        115 & 85.93134179 & 50.61006225 & 0 & 0  \\ \hline
        116 & 85.93134179 & 50.61006225 & 10.23733776 & 0  \\ \hline
        117 & 85.93134179 & 50.61006225 & 0 & 0  \\ \hline
        118 & 85.93134179 & 50.61006225 & 18.14882145 & 0  \\ \hline
    \end{tabular}
		\caption{Ex.~\ref{ex-c} Energy Table of Case 1: Non-modulated Waveform}
		\label{table:exccase1}
\end{table}

\begin{table}[!ht]
		\centering
    \begin{tabular}{|l|l|l|l|l|l|l|l|l|l|}
    \hline
        Frames & Lobe A & Lobe B & Lobe C & Lobe D  \\ \hline
        100 & 76.23429076 & 28.33540646 & 21.84815468 & 14.94586226  \\ \hline
        101 & 76.23429076 & 0 & 0 & 0  \\ \hline
        102 & 76.23429076 & 28.33540646 & 0 & 0  \\ \hline
        103 & 76.23429076 & 28.33540646 & 21.84815468 & 14.94586226  \\ \hline
        104 & 76.23429076 & 28.33540646 & 21.84815468 & 14.94586226  \\ \hline
        105 & 76.23429076 & 28.33540646 & 21.84815468 & 14.94586226  \\ \hline
    \end{tabular}
		\caption{Ex.~\ref{ex-c} Energy Table of Case 2: Modulated Waveform}
		\label{table:exccase2}
\end{table}

In conclusion, there are multiple evidence from the experiments showing that by applying modulating algorithm discussed in this paper, the modulated unloaded waveform of vibration system in IR video is able to predict the performance of loaded waveform of vibration system. This modulating algorithm can put in used for payload prediction in transportation field for variety research purpose.

\end{appendix}

\bibliographystyle{amsplain}
\bibliography{NSrefs}

\end{document}